\documentclass[12pt]{article}
\usepackage[inner=1in,outer=1in,top=1in,bottom=1in,marginparwidth=.7in,marginparsep=.1in]{geometry}
\usepackage{mathtools} 
\usepackage{graphicx,multirow}
\usepackage{amssymb,latexsym,amsmath,amsthm,mathrsfs}
\usepackage{color}
\usepackage{tikz}
\usepackage[all]{xy} 
\usepackage{caption}
\usepackage{subcaption}
\usepackage{url}
\numberwithin{equation}{section}
\theoremstyle{plain}
\newtheorem{theorem}{Theorem}[section]

\newtheorem{lemma}[theorem]{Lemma}

\newtheorem{proposition}[theorem]{Proposition}

 \newtheorem*{thm*}{Theorem}

\theoremstyle{definition} 
\newtheorem{remark}[theorem]{Remark} 

\newtheorem{definition}[theorem]{Definition}

 \newtheorem{example}[theorem]{Example}

\newcommand{\K}{\mathbb{K}}
\newcommand{\Z}{\mathbb{Z}}
\newcommand{\R}{\mathbb{R}}
\newcommand{\C}{\mathbb{C}}
\renewcommand{\a}{\alpha}
\renewcommand{\b}{\beta}

\title{Finite Alphabet Phase Retrieval}
\author{Tamir Bendory, Dan Edidin, Ivan Gonzalez}

\begin{document}
\maketitle
\begin{abstract}
We consider the finite alphabet phase retrieval problem: recovering a signal whose
  entries lie in a small alphabet of possible values from its Fourier magnitudes. 
  This problem arises in the celebrated technology of X-ray crystallography to determine the atomic structure of biological molecules. 
   Our main result states that for generic values of the alphabet,
 two signals have the same Fourier magnitudes if and only if several partitions have the same difference sets. Thus, the finite alphabet phase retrieval problem reduces to the combinatorial problem of determining a signal from those difference sets.
   Notably, this result holds true when one of the letters of the alphabet is zero, namely, for sparse signals with finite alphabet, which is the situation in X-ray crystallography.   
\end{abstract}

\section{Introduction}

X-ray crystallography is a leading technology for determining the 3-D atomic structure of biological molecules, such as proteins. 
Indeed, thousands of  new crystal structures are resolved each year, and  more than a dozen Nobel Prizes have been awarded for work involving
X-ray crystallography. 
In X-ray crystallography, the
crystal---a periodic arrangement of a repeating unit---is illuminated with a beam of X-rays, producing a diffraction pattern, which is  equivalent to the   magnitude of the Fourier transform of the  crystal.
The signal to be estimated (the electron density function of the crystal)  is supported  only at the sparsely-spread positions of atoms~\cite{millane1990phase}.
Therefore, the {\em crystallographic phase retrieval problem} entails recovering a sparse signal from its Fourier magnitudes. 
The crystallographic phase retrieval is a special case of the {\em phase retrieval} problem, which refers to all problems that involve recovering a signal from its Fourier magnitudes, see~\cite{shechtman2015phase,bendory2017fourier,grohs2020phase,bendory2022algebraic} and reference therein. 
A detailed mathematical model of X-ray crystallography is introduced
in~\cite{elser2018benchmark}.

A recent paper by a subset of the authors provides the  first rigorous attempt to establish a mathematical theory for the crystallographic phase retrieval problem~\cite{bendory2020toward}.  
In particular, it was conjectured that a generic sparse signal $x\in\R^N$ whose support has size at most $K$ is uniquely determined, up to unavoidable ambiguities, as long as $K\leq N/2$.
The conjecture was verified for a small set of parameters; see also~\cite{ghosh2022sparse,ranieri2013phase}.

 In practice, however, a more accurate model of the crystallographic phase retrieval problem  should account
	for sparse signals whose non-zero entries are taken from a finite (small) alphabet; this alphabet models the relevant type of atoms, such as hydrogen, oxygen, carbon,
	nitrogen, and so on. 
In this paper, we make a first step towards this direction.  
Specifically, we study the problem of recovering a {discrete one-dimensional periodic} signal, whose entries are taken from a finite alphabet, from its Fourier magnitudes. We refer to this problem as the {\em finite alphabet phase retrieval problem}.
{We note that  recovering problems of finite alphabet signals were studied before, but mostly under linear models~\cite{keiper2017compressed,keiper2017reconstruction,sarangi2021no,sarangi2023super}.}

In particular, we show  that for generic choice of entries in the alphabet, 
the finite alphabet phase retrieval problem can be reduced to a combinatorial problem
involving difference sets. This is similar to the situation for binary
phase retrieval (a problem studied before~\cite{elser2017complexity}) but new combinatorial subtleties can occur.
More specifically, we show that two signals with entries taken
from a finite alphabet have the same Fourier magnitudes if and only if the associated partitions have the same difference sets;
see Proposition~\ref{prop:main}.  
Notably, this result remains true when one of the letters of the alphabet is zero, namely, for sparse signals with finite alphabet; see Theorem~\ref{thm:sparse}. 
 This is the situation in X-ray crystallography where crystals are typically very sparse; the non-zero values occupy only $\sim 1/100$ of the signal's support~\cite{elser2017complexity}.
Unfortunately, the problem of analyzing if specific difference sets determine a set uniquely (up to unavoidable symmetries) is an extremely difficult combinatorial problem \cite[p.350]{rosenblatt1982structure}, \cite[Section 3]{rosenblatt1984phase}. 
Therefore, we cannot provide a complete characterization when a finite-alphabet signal can be recovered uniquely, up to unavoidable symmetries, from its  Fourier magnitudes.

\begin{remark}
  {In this paper we follow a long tradition in the crystallography literature and restrict our discussion to the one-dimensional phase retrieval problem for periodic signals. This corresponds to viewing our signals as functions on
    the cyclic group $\Z_N$. As was the case in \cite{bendory2020toward}, much of our theory can be readily adapted to study functions on any abelian group such as $\Z_N \times \Z_N$; see Section \ref{sec.abelian}.}
\end{remark}

The rest of the paper is organized as follows. Section~\ref{sec:background} provides a necessary background on difference sets, homometric sets, and autocorrelations.  
Section~\ref{sec:binary} begins our analysis by studying signals whose entries are taken from two-letters alphabet. Section~\ref{sec:finite} presents and proves our main results about phase retrieval of signals whose entries are taken from a finite alphabet. Section~\ref{sec:examples} provides a few examples and introduces the intriguing notion of pseudo-equivalent partitions. {Section~\ref{sec.extension} introduces a few possible directions for future research on this problem.}

\section{Background} \label{sec:background}
We begin by introducing basic definitions about difference sets. 
For any $i,j \in [0,N-1]$, we define  the {\em cyclic distance} between $i$ and $j$ by 
\begin{equation}
	d(i,j) = \min\{N - |i-j|, |i-j|\}. 
\end{equation}
We note that $d(i,j) \in [0,\lfloor N/2\rfloor]$ as illustrated in
Figure \ref{fig.distance}.
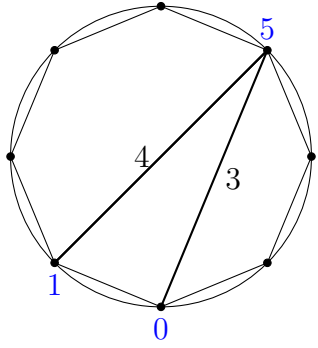
\begin{figure}
\def\n{8} 
\def\r{2}
	\begin{center}\begin{tikzpicture}[
		dot/.style={draw,fill,circle,inner sep=1pt}
		]
		\draw (0,0) circle (\r);
		
		\foreach \i in {1,...,\n} {
			\node[dot] (w\i) at (\i*360/\n:\r) {}; 
		}
		\foreach \i [evaluate=\i as \j using {int(\i+1)}] in {1,...,\n} {
			\draw[-] (\i*360/\n:\r) -- (\j*360/\n:\r);	
	} 

	\draw[-,thick] (5*360/\n:\r) -- (1*360/\n:\r) node[above,blue]{5}  node[midway,left]{4};
	\draw[-,thick] (1*360/\n:\r) -- (5*360/\n:\r) node[below,blue ]{1} ;
	\draw[-,thick] (1*360/\n:\r) -- (6*360/\n:\r) node[below, blue]{0} node[midway,right]{3};
	\end{tikzpicture}
	\end{center}
\caption{Illustration of the cyclic distances between the points $0$ and $5$ and $1$ and $5$ in $[0,7]$.}\label{fig.distance}
\end{figure}

\begin{definition}[Difference sets]
  Let $A, B$ be two subsets of $[0,N-1]$. 
  We define the
  cyclic difference multi-set by
  \begin{equation}
  	A-B = \{ d(i,j) \mid i \in A, 
  	j \in B\}. 
\end{equation}
In particular, the self difference multi-set is given by
 $A-A = \{ d(i,j) \mid i \leq j
\in A\}$.
\end{definition}

\begin{definition}[The dihedral group]
  The dihedral group $D_{2N}$ is the group of symmetries of the regular
  $N$-gon. It is a group of order $2N$, which is generated by two elements $r$ (rotation) and $s$ (reflection). The elements $r, s$ satisfy the relations
  $r^N = s^2 = e$ and $rs = sr^{N-1}$, where $e$ is the identity.
  The group $D_{2N}$ acts on the set $[0,N-1]$ as follows. The element
  $r \in D_{2N}$ acts by cyclic shift; i.e., $r(i) = (i+1)\bmod N$ and
  the element $s$ acts by the reflection $s(i)= N-i$.
\end{definition}

Next, we define {equivalence} classes. Since the Fourier magnitude is invariant under cyclic shifts and reflection (i.e., under the dihedral group), we can only hope to determine the support of a signal from the signal's Fourier magnitudes up to an action of the dihedral group $D_{2N}$. 

\begin{definition}
	Two subsets $A, B \subset [0,N-1]$ are {\em equivalent} if there exists an element $\sigma \in D_{2N}$ such that $B = \sigma A$.
\end{definition}

Another fundamental definition is that of homometric sets: subsets with the same difference set.

\begin{definition}
  Two subsets $A, B$ are {\em homometric} if $A-A = B-B$.
\end{definition}

\begin{lemma} \label{lem:homometric}
  If $A, B$ are equivalent then they are homometric.
\end{lemma}

\begin{proof}
  The lemma is an immediate consequence of the fact that the cyclic distance $d(i,j)$ is invariant under cyclic shifts
  and reflection.
\end{proof}

A key fact we use about homometric subsets of $[0,N-1]$ is the following
result, originally stated by Patterson \cite{patterson1944ambiguities}. For a modern proof, see \cite[Corollary 1]{iglesias1981pattersons} or \cite{buerger1976proofs}.

\begin{theorem}[Patterson] \label{thm.patterson}
  Two sets $ A,B \subset [0,N-1]$ are homometric if and only
  if their complements are homometric as well.
\end{theorem}

We now consider the case of partitions. Let $A_1, \ldots, A_K$
and $B_1, \ldots, B_K$ be two ordered partitions of $[0,N-1]$.
\begin{definition}
  Two ordered partitions  $A_1, \ldots , A_K$ and $B_1, \ldots , B_K$ are
  {\em homometric} if $A_i -A_j = B_i -B_j$ for all pairs $i,j \in \{1,\ldots,K\}$.
  Two ordered partitions  $A_1, \ldots , A_K$ and $B_1, \ldots , B_K$ are
  {\em equivalent} if there exists
  $\sigma \in D_{2N}$ such that $B_i = \sigma A_i$ for all $i \in \{1,\ldots,K\}$.
\end{definition}

Lemma~\ref{lem:homometric} can be directly extended to ordered partitions.
\begin{lemma} 
  Equivalent partitions are homometric.
\end{lemma}

Before moving to the next section, we remind the reader of a couple of definitions from signal processing. 

\begin{definition}[Power spectrum and periodic autocorrelation] The power
	spectrum of a signal $x \in \C^N$ is  the vector $|\hat{x}|^2\in\R^N_{\geq 0}$,  where $\hat{x}$ is the discrete Fourier transform (DFT) of $x$ and the
  absolute value is taken componentwise.
  The periodic  auto-correlation of $x$ is defined by 
  \begin{equation}
  	a_x[\ell] = 	\sum_{n=0}^{N-1}x[n]\overline{x[\ell + n]},
  \end{equation}
  where all indices
  are taken modulo $N$.
\end{definition}

A key fact, first observed by Patterson \cite{patterson1934fourier,
  patterson1935direct}, is that
the DFT of $a_x$ is the power spectrum. 
The phase retrieval problem is thus equivalent to the problem of recovering
a signal from its periodic auto-correlation.
In the sequel, we use the terms Fourier magnitudes, power spectrum and autocorrelation interchangeably.

\section{Binary and two-alphabet phase retrieval} \label{sec:binary}
The binary phase retrieval problem is the problem of recovering a
binary signal $ x \in \R^N $ from its periodic auto-correlation $ a_x \in \R^N$~\cite{elser2017complexity,bendory2022sparse}.
Let $S(x)$ denote the support of a signal $x$. 
A well known result for binary signals states that  $a_x = a_{x'}$ if and only if $S(x) - S(x) = S(x') - S(x')$; i.e.,
the two supports are homometric (e.g.,~\cite{bendory2020toward}).

Let us expand upon this result and  replace the zeros and ones in the traditional binary phase retrieval with
arbitrary scalars $\a$  and $\b$.
 We refer to this problem
as the {\em  two-alphabet phase retrieval problem}. When either $\a$ or
$\b$ is zero, then this problem reduces to the binary phase retrieval problem.
Otherwise, we do not have a well-defined notion of a support set since neither one of the two letters is assumed to be zero. Instead, we consider the sets
\[ S_\a(x) = \{i \mid x[i] = \a \} \qquad\qquad\text{and}\qquad\qquad S_\b(x) = \{j \mid x[j] = \b\}.\]

We prove the following result. 
\begin{theorem}\label{theorem:1}
  For a generic choice of values of $\a,\b$, the following
  are equivalent for signals $ x,x' \in \{\a,\b\}^N $.
  \begin{description}
\item{(i)} $a_x = a_{x'}$
\item{(ii)} $S_\a(x)$ and $S_\a(x')$ are homometric
\item{(iii)} $S_\b(x)$ and $S_\b(x')$ are homometric
\item{(iv)} The ordered partitions $(S_\a(x),S_\b(x))$ and $(S_\a(x'), S_\b(x'))$are homometric.
  \end{description}
\end{theorem}
\begin{remark}
  {By generic choice of $\a,\b$, we mean that the set of $\a,\b$ for which the conclusion of the theorem does not hold is contained in the zero set
    of a collection of non-zero polynomials in $\R[\a,\b]$. In particular, the pairs $(\a,\b) \in \R^2$ for which the theorem holds has full Lesbegue measure.}
\end{remark}

\begin{proof}
  Clearly $(iv)$ implies $(ii)$ and $(iii)$. If we show that
  $(ii)$ implies $(iv)$ then by symmetry we can also 
  conclude that $(iii)$ implies $(iv)$. To see that $(ii)$ implies
  $(iv)$ we use Patterson's Theorem, Theorem~\ref{thm.patterson}. Note that $S_\b(x) = S_\a(x)^{c}$,
  so by Patterson's Theorem  $S_\a(x)$ and $S_\a(x')$ are homometric if and only if $S_\b(x)$ and $S_\b(x')$ are also homometric. To show
  that the partitions $(S_\a(x), S_\b(x))$ and $(S_\a(x'), S_\b(x'))$
  are homometric we must show that $S_\a(x) - S_\b(x) = S_\a(x') -S_\b(x')$.
  This follows from the fact that the difference sets $S_\a(x) - S_\a(x), S_\a(x) - S_\b(x), S_\b(x) - S_\b(x)$ (respectively $S_\a(x') - S_\a(x'), S_\a(x') - S_\b(x'),
  S_\b(x') - S_\b(x')$) form a partition of the multi-set
  $[0,N-1] - [0,N-1]$.

 Let $a_x[\ell]$ and $a_{x'}[\ell]$ be the $\ell$-th entry of $a_x$ and $a_{x'}$, respectively. Then, 
  $a_x[\ell] = m_\ell \alpha^2 + n_\ell \beta^2 + p_\ell \alpha \beta$,
  and $a_{x'}[\ell] = m'_\ell \alpha^2 + n'_\ell \beta^2 + p'_\ell \alpha \beta$, 
  where $m_\ell$ (resp. $m'_\ell$) is the multiplicity of $\ell$ in $S_\alpha(x) - S_\alpha(x)$ (resp. $S_\alpha(x') - S_\alpha(x')$),
$n_\ell$ (resp. $n'_\ell$) is the multiplicity of $\ell$ in $S_\beta(x) - S_\beta(x)$ (resp. $S_\beta(x') - S_\beta(x')$),
  and
  $p_\ell$ (resp. $p'_\ell$) is the multiplicity of $\ell$ in $S_\alpha(x) - S_\beta(x)$ (resp. $S_\alpha(x') - S_\beta(x')$).
  Hence, if $(S_\alpha(x), S_\beta(x))$ and $(S_\alpha(x'), S_\beta(x'))$
  are homometric then $a_x[\ell] = a_{x'}[\ell]$. Thus $(iv) \implies (i)$.

  Conversely, suppose that $a_x = a_{x'}$. This means
  that for each $\ell$ we have that
  $$m_\ell \alpha^2 + n_\ell \beta^2 + p_\ell \alpha \beta = m'_\ell \alpha^2 + n'_\ell \beta^2 + p'_\ell \alpha \beta.$$
  Let us define 
  $m = m_\ell - m'_\ell$, $n = n_\ell - n'_\ell$, $p = p_\ell - p'_\ell$.
  {If the integers $m,n,p$ are not all zero, then $\alpha, \beta$
    must be contained in the zero set of a quadratic polynomial
    with integer coefficients $m,n,p$, each of absolute value at most
    $N$. This means that for a generic choice of $\a,\b$ we must have
    that $m =n=p=0$;
    }
  i.e., $m_\ell = m'_\ell$, $n_\ell = n'_\ell$, and $p_\ell = p'_\ell$, meaning that the multiplicities of $\ell$ in $S_\alpha(x) - S_\alpha(x),
  S_\beta(x) - S_\beta(x), S_\alpha(x) - S_\beta(x)$ are equal
  to the multiplicities of $\ell$ in $S_\alpha(x') - S_\alpha(x'),
  S_\beta(x') - S_\beta(x'), S_\alpha(x') - S_\beta(x')$, respectively.
  Since for generic $\alpha, \beta$ this is true for every $\ell = 0,\ldots , N-1$ we conclude that
  the partitions $(S_\alpha(x), S_\beta(x))$ and $(S_\alpha(x'), S_\beta(x'))$
  are homometric. Thus, $(i) \implies (iv)$ for generic choice of $\alpha, \beta$.
  
 \end{proof}

\section{Signals with entries taken from a finite alphabet}\label{sec:finite}

We now extend our analysis to account for signals whose entries are taken from a finite alphabet. 
Let $S = \{\a_1, \ldots , \a_{K}\}$ be a set $K$ real numbers
and  let 
$\R_S$ be the set of all vectors in $\R^N$ whose entries are taken from
$S$. A vector $x \in \R_S$ determines a length $K$ partition $A_1(x), \ldots A_K(x)$ of $[0,N-1]$, 
where $A_k(x)  = \{n \in [0,N-1]\,| \, x[n] = \a_k\}$.

\begin{remark}
  In this paper we assume that the alphabet is taken from the reals, but the theory is unchanged if we consider complex alphabet entries.
\end{remark}

\begin{proposition} \label{prop:main}
  For a generic choice of $\a_1, \ldots , \a_{K}$, two vectors
  $x,x' \in \R_S$ have the same auto-correlation if and only if
  the associated partitions $\{A_k(x)\}$ and $\{A_k(x')\}$ are homometric.
\end{proposition}
\begin{proof}
  The proof is similar to the proof of the equivalence $(i) \iff (iv)$ in
  Theorem \ref{theorem:1}.
  
  Since the entries $\a_1, \ldots \a_K$ are generic we can treat them as indeterminates. 
  By definition, $a_x[\ell] = \sum_{n =0}^{N-1} x[n]x[n+ \ell]$
  and $a_{x'}[\ell] = \sum_{n=0}^{N-1} x'[n] x'[n+\ell]$, where
  all indices are taken modulo $N$. Since the entries
  of $x,x'$ are taken from the set $S$, 
  $a_x[\ell]$ and $a_{x'}[\ell]$ are quadratic polynomials in $\a_1, \ldots ,\a_K$. The coefficient of $\a_i \a_j$ in $a_x[\ell]$ is the multiplicity of~$\ell$
  in the difference multi-set $A_i(x) - A_j(x)$. Likewise, 
  the coefficient of $\a_i \a_j$ in $a_{x'}[\ell]$ is the multiplicity of $\ell$
  in the difference multi-set $A_i(x') - A_j(x')$. Hence,
  $a_x = a_{x'}$ if and only $A_i(x) - A_j(x) = A_i(x') - A_j(x')$
  for all $i,j$. In other words, $a_x = a_{x'}$ 
  if and only if the corresponding partitions
  are homometric.
\end{proof}

We are now ready  to present our main result, which is motivated by the crystallographic phase retrieval problem of recovering a sparse signal, whose non-zero values are taken from a finite alphabet, from its autocorrelation. 

\begin{theorem}[Sparse signals] \label{thm:sparse}
  Consider two signals $x,x'$ with entries taken from an alphabet
  ${0,\a_2, \ldots , \a_K}$ {with $\a_2,\ldots, \a_K$ generic}. 
  Then, $a_x = a_{x'}$ if and only if the partitions $A(x)$ and $A(x')$
  are homometric.
\end{theorem}

\begin{remark}
  The significance of this result is that we no longer assume that
  $\a_1$ is arbitrary. Equivalently, our result states that if we view
  $a_x, a_{x'}$ as polynomials in the variables $\a_1, \ldots ,\a_K$,
  then $a_x(\a_1, \ldots , \a_K) = a_{x'}(\a_1, \ldots , \a_K)$ if and only
  if $a_x(0, \a_2, \ldots , \a_k)= a_{x'}(0, \a_2, \ldots , \a_K)$.
\end{remark}

\begin{proof}
  {If the partitions $A(x)$ and $A(x')$ are homometric then clearly
    $a_x =a_{x'}$.}

  {Conversely,} if $a_x = a_{x'}$, then for every $\ell$ the coefficients
  of $\a_i \a_j$ in $a_x[\ell]$ and $a_{x'}[\ell]$ are equal
  for $i, j \geq 2$. As before, this coefficient
  is just the multiplicity of $\ell$ in the difference multi-sets
  $A_i(x) - A_j(x)$ and $A_i(x') - A_j(x')$ respectively. {Using the
  same reasoning as in the proof of Theorem \ref{theorem:1} 
  we see that if $\a_2, \ldots , \a_k$ are generic then
  $A_i(x) - A_j(x) = A_i(x') -A_j(x')$ for $i,j> 1$.}
  

We begin with the following lemma. 

  \begin{lemma} \label{lem.pigeonhole}
    $A_1[x] - A_1[x] = A_1[x'] - A_1[x']$.
  \end{lemma}
  \begin{proof}
    Let $B[x] = \cup_{i = 2}^K A_i[x]$ and $B[x'] = \cup_{i=2}^K A_i[x']$.
    Then, $\{A_1[x], B[x]\}$ and $\{A_1[x], B[x']\}$ are length two partitions of $[0,N-1]$.
    Now, $B[x] - B[x] = \uplus_{i,j \geq 2}  (A_i[x] - A_j[x]) =B[x'] -
    B[x']$. {(Here, the notation $\uplus $ refers to the additive union of
    multi-sets. If
    $C_1, \ldots , C_r$ are multi-sets, then an element $c \in \uplus C_i$ appears
    with multiplicity equal to the sum of the multiplicities in each of
    the sets $C_i$.)}
    In other words, the subsets $B[x]$ and $B[x']$ are homometric. 
    Since $A_1[x] = B[x]^{c}$ and $A_1[x'] = B[x']^{c}$, it follows from
    Patterson's Theorem that they also have the same difference sets.
  \end{proof}

  To complete the proof that the partitions are homometric, we need to show
  that the difference sets $A_1[x] - A_i[x]$ and $A_1[x'] - A_i[x']$
  are equal for all $i > 1$. To do this we can argue inductively.
  Assume by induction that $A_1[x] - A_i[x] = A_1[x'] - A_i[x']$
  for $i < k$ with the initial case $k =1$ established by Lemma \ref{lem.pigeonhole}.
  Let $B_{k}[x] = \cup_{j > k} A_k[x]$ and $B_{k}[x'] = \cup_{j > k} A_k[x']$.
  We know already that $B_k[x] - B_k[x] = B_k[x'] - B_k[x']$.
  Hence, by Patterson's Theorem we know that
  $B_k[x]^{c} - B_k[x]^{c} = B_k[x']^{c} - B_k[x']^{c}$. But these
  difference sets are just
  $\uplus_{i, j \leq k}  A_i[x] - A_j[x]$ and $\uplus_{i,j \leq k} A_i[x'] - A_j[x']$. We a priori know that if $i,j \geq 2$ then  
  $A_i[x] - A_j[x] = A_i[x'] - A_j[x']$. By induction, we know that
  $A_1[x] - A_j[x] = A_1[x'] - A_j[x']$ for $j < k$. Hence, by the pigeon-hole principle we conclude that $A_1[x] - A_k[x] = A_1[x] - A_k[x']$ for all $k$.
\end{proof}

\section{Equivalent, Homometric and pseudo-equivalent partitions} \label{sec:examples}
In the conclusion to his 1944 paper \cite{patterson1944ambiguities},
Patterson noted that 
\begin{quote} 
	in very few cases
  are the atoms of a crystal all of one kind and it seems very probable that the presence of a second kind of atom will often resolve the ambiguities which might occur in the location of the first if taken alone.
  \end{quote}
In the following example we illustrate this phenomenon for 1-D signals.

\begin{example} \label{ex:partitions}
  Consider $A= \{0,1,4,7\}$ and $A' = \{0,1,3,4\}$ as subsets of $[0,7]$. These sets are homometric but not equivalent. Let $B = A^{c}= \{2,3,5,6\}$ and
  $B' = (A')^{c} =\{2,5,6,7\}$. Direct inspection or Patterson's Theorem
  implies that the sets $B$ and $B'$ are also homometric but not equivalent. In particular $B-B = B'-B' = \{0^4, 1^2,2^2,3,4\}$ as multi-sets and
any two binary signals supported on $B, B'$ have the same autocorrelation. 
  
  Now, we consider the decomposition of $B$ and $B'$ into two subsets
  of size two $B_1, B_2$ and $B_1', B_2'$, respectively.
  We can now ask which of the six possible ordered partitions $A, B_1, B_2$
  are homometric with any of the six possible partitions $A', B_1', B_2'$.
 Of the six possible three-set partitions of the form $A, B_1, B_2$, 
 only the partitions $(\{0,1,3,4\}, \{2,6\}, \{3,5\})$ are homometric but not
 equivalent to
 the partitions $A', B'_1, B_2'= (\{0,1,4,7\}, \{2,6\}, \{5,7\})$.
 For example the partitions $(\{2,3\}, \{5,6\})$ and $(\{2,5\}, \{6,7\})$
 are not homometric as illustrated in Figure \ref{fig.nonhomo}.
 
 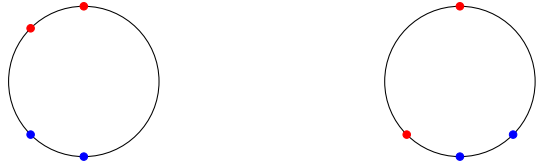
\begin{figure}
    \def\n{8} 
\def\r{1} 

\begin{center}	\begin{tikzpicture}[
		dot/.style={draw,fill,circle,inner sep=1pt}
		]
		\draw (0,0) circle (\r); 
		\draw (5.0,0) circle (\r);
		
		\foreach \i in {2,3 }{
			\node[dot,red] (w\i) at (\i*360/\n:\r) {};	
		} 
		\foreach \i in {5,6} {
			\node[dot,blue] (w\i) at (\i*360/\n:\r) {};	
		} 
		\foreach \i in {2,5} {
			\node[dot,red] (w\i) at ({\r*cos(\i*360/\n)+5.0},{\r*sin(\i*360/\n)}) {};
		} 
		\foreach \i in {6,7} {
			\node[dot,blue] (w\i) at ({\r*cos(\i*360/\n)+5.0},{\r*sin(\i*360/\n)}) {};
		} 
		
  \end{tikzpicture}
  \end{center}
		\caption{{\small The ordered partition $(B_1, B_2) = (\{2,3\},\{5,6\})$ is shown on the left circle with points of
		$B_1$ in red and $B_2$ in blue. The ordered partition $(B'_1, B'_2) = (\{2,5\},\{6,7\})$
		is shown on the right circle with the points in $B'_1$ in red and $B'_2$ in blue. These 
		partitions are clearly not homometric since the distance between the red points on the left is one but in the circle on the right the distance is three. However, the sets $B_1 \cup B_2 = \{2,3,5,6\}$ and $B_1' \cup B_2' = \{2,5,6,7\}$ are homometric; see Example~\ref{ex:partitions}.}} \label{fig.nonhomo}
		\end{figure}
 On the other hand, if we consider partitions of $B$ and $B'$ into three subsets,
 respectively, then no partition of the form $A,B_1, B_2, B_3$ is homometric with
 a partition of the form $A, B'_1, B'_2, B'_3$.
This reinforces Patterson's intuition that by considering a second kind of atom
 we increase the likelihood that the ambiguities of the auto-correlation can be resolved. 

\end{example}

We next consider  pseudo-equivalent partitions. As far as we know, this case was not considered before.

\begin{definition}
  Two partitions $A_1,\ldots , A_K$ and $A'_1, \ldots , A'_K$ of $[0,N-1]$
  are {\em pseudo-equivalent} if there exists elements $\sigma_1, \ldots , \sigma_K \in D_{2N}$
  such that $A'_i = \sigma_i A_i$
\end{definition}

When $K=2$, any pair of pseudo-equivalent partitions are automatically equivalent because in this case if $\sigma_1 A_1 = A_1'$ then $\sigma_1 A_2 = A_2'$ since
$A_2 = A_1^{c}$ and $A_2' = A_1^{'c}$, respectively.
{Note that if the partitions $A_1, \ldots ,A_k$ and $A'_1, \ldots , A'_k$
  are pseudo-equivalent then $A_i -A_i = A'_i -A'_i$ for each $i$.}
However, in general pseudo-equivalent partitions need not be homometric and homometric
partitions need not be pseudo-equivalent.

\begin{example} 
  The ordered partitions $\{0,1,4\}, \{7\}, \{3\}, \{2,5,6\}$ and
  $\{0,1,4\}, \{3\}, \{7\}, \{2,5,6\}$ are pseudo-equivalent and also
  homometric but not equivalent.
\end{example}

{Although homometric partitions need not be pseudo-equivalent and vice-versa,  a numerical experiment seems to indicate that for approximately uniform partitions (i.e., partitions where the sets have approximately the same size)}
homometric partitions are in fact pseudo-equivalent.

\begin{example} \label{ex:pseudo}
  We conducted the following experiment. For a given $N$ in the range, $N = 6, \ldots, 13$, we considered a set $S(N)$ of partitions
where $N_1 = \lceil N/3 \rceil$, $N_2 = \lceil N - N_1/2 \rceil $, $N_3 = N -N_2 -N_3$. When $N=6,7$ we considered
  all partitions and when $N \geq 8$ we took a random sample of size
  300. (For $N =6,7$ there are fewer than 300 partitions).
Each of the $\binom{|S(N)|}{2}$ pairs of partitions in $S(N)$ were
tested to see if they were homometric.
As indicated in Table~\ref{table}, except for the case $N =12$, all homometric pairs found were either equivalent or pseudo-equivalent.

 \begin{table}
  \begin{tabular}{c| c| c| c| c|}
    N & Partition Sizes &Equivalent Pairs & Pseudo-equivalent pairs & Total homometric pairs\\
\hline
6 & 2,2,2,& 369 &  0 &  369\\
7 & 3,2,2 & 1218 & 0 & 1218\\
8 & 3,3,2 & 2005 & 99 & 2104\\
9 & 3,3,3 &  813 & 158 & 971\\
10 & 4,3,3 & 360 & 12 & 372\\
11 & 4,4,3 &148 &1 &149\\
12 & 4,4,4  & 62&3 &70\\
13 & 5,4,4 & 10& 0& 10
    \end{tabular}
\caption{Results of Example~\ref{ex:pseudo} \label{table}}
  \end{table}

\end{example}

Unfortunately, the problem of analyzing, for a particular partition type, which partitions are homometric but not equivalent seems to be an extremely difficult combinatorial problem. 
We conclude with the following positive result when each atom in the alphabet appears with multiplicity exactly one. {To state our result we introduce the following notation. An ordered partition of $A_1, \ldots , A_K$ of $[1,N]$
  has {\em type}
  $[n_1, \ldots , n_K]$, where $n_i = |A_i|$.}
\begin{proposition}
Any  two ordered partitions of {type} $[N-K,\underbrace{1,\ldots,1}_{K \text{times}}]$ are homometric
if and only if they are equivalent. 
\end{proposition}
\begin{proof}
  We use induction on $K$. For $K=1$, any two partitions of type $[N-1,1]$
  are necessarily equivalent because the dihedral group acts transitively on
  the set $[0,N-1]$. Hence, if $(A, \{\a\})$ and $(A', \{\a'\})$ are two partitions
  of this form, then there exists $\sigma \in D_{2N}$ such that $\a' = \sigma(\a)$.
 
  {Assume by induction that the statement holds for partitions
    of type $[N-K, 1, \ldots , 1]$; i.e.,  
if $(A',\{a'_1\},\ldots , \{a'_K\})$ and $(A,\{a_1\}, \ldots
, \{a_K\})$ are two homometric ordered partitions, then they are equivalent.
We will prove that the statement holds for partitions of type
  $[N-K-1,1,\ldots, 1]$.}

  {The induction hypothesis implies
  that if  $(\a_1', \ldots , \a_K')$ and $(a_1, \ldots , a_K)$ are
  any sequences of distinct integers in $[0,N-1]$ such that
$d(a'_i, a'_j) = d(a_i,a_j)$}
  for all $1 \leq i < j \leq K$
  then there exists $\sigma \in D_{2N}$ such that $\a_i'= \sigma \a_i$
  for $i = 1, \ldots , K$.  To establish the induction step we must prove
  that if $(\a_1, \ldots , \a_K, \a_{K+1})$ and $(\a'_1,\ldots , \a'_K, \a'_{K+1})$
  are two sequences 
{such that $d(a'_i, a'_j) = d(a_i, a_j)$ 
for $1 \leq i, j \leq K+1$},
then there exists $\sigma \in D_{2N}$ such that $\a_i' = \sigma \a_i$ for $i =1,\ldots K+1$. 
  By induction, there exists $\sigma_1 \in D_{2N}$ such
  that $\sigma_1(\a_1, \ldots , \a_K) = (\a_1', \ldots , \a_K')$.
  In particular, we may assume that $(\a_1', \ldots , \a_K',\a_{K+1}')$ is equivalent to a sequence of the form $(\a_1, \ldots , \a_K, \a_{K+1}')$.
  Applying a suitable element of $D_{2N}$ we may also assume that $\a_1 = 0$.
  Hence, $d(\a_{K+1},0)= d(\a'_{K+1},0)$
  which implies that
  $\a'_{K+1} = N- \a_{K+1}$ or $\a'_{K+1} = \a_{K+1}$. In the latter case we are done since the sequences would be equal.
  If $\a'_{K+1} = N-\a_{K}$ then for $i =2, \ldots , K$ we have that  
  $d(N - \a_{K+1}, \a_i) = d(\a_{K+1},\a_i)$
  This implies that
  either $2\a_{K+1} \equiv 0 \bmod N$ or $2\a_i \equiv 0 \bmod N$. Since neither
  $\a_{K+1}$ nor $\a_i$ can be zero, we see that in either case $N$ must be even and in the former case we have that $\a_{K+1} = N/2$ so
  $\a'_{K+1} = \a_{K+1}$. On the other hand, the integers $\a_2, \ldots , \a_K$
  are distinct so we cannot have that $\a_i = N/2$ for $2 \leq i \leq K$ unless
  $K=2$. In the case of $K=2$, then we see that our sequences would necessarily be
  of the form
  $(0,N/2, \a_{3})$ and $(0,N/2, N-\a_{3})$. However, these sequences are also dihedrally equivalent since they are related by the reflection $\a \mapsto N-\a$.
\end{proof}
\begin{remark}
  {When $N-K > N/2$ we expect that analogous results hold for
    partitions of type $[N-K,a_1,\ldots ,a_L]$ where $L \ll K$ and the
    $a_\ell$ are approximately equal. Unfortunately, investigating
    this problem is currently beyond reach from both a computational
    and theoretical perspective.}
\end{remark}

\section{Extensions} \label{sec.extension}
{
\subsection{Finite alphabet phase retrieval in finite abelian groups} \label{sec.abelian}
The finite alphabet phase retrieval problem  can be generalized to
any finite abelian group. Let $G$ be a finite abelian group
and let $V$ be the vector space of functions $x \colon A \to \K$, where
$\K = \R$ or $\K=\C$.
In the case of one-dimensional phase retrieval, $G =\Z_N$, and for higher-dimensional phase retrieval $G = \Z_N^M$ is a product of cyclic groups.  In this case,
the auto-correlation is defined as a function $A \to \K$ defined by the formula
\[ a_x[\ell] = \sum_{\ell' \in A} x[\ell']\overline{x[\ell + \ell']}.\]

Given a subset $A \subset G$, we can again define the $G$-difference set $A-A$
\cite[Appendix E]{bendory2020toward},
and define the notion of homometric sets. The $G$-difference set 
is invariant under the action of a group $D_G = G \ltimes \Z_2$, 
and we say that two subsets $A,A'$ are equivalent if there is an element
$\sigma \in D_G$ such that $A' = \sigma A$. The proof of Lemma \ref{lem:homometric} easily generalizes to show that
two equivalent subsets of $G$ are homometric. Likewise,
 Theorem~\ref{theorem:1} and Proposition~\ref{prop:main} generalize to partitions of finite abelian groups. However, the existing proofs Patterson's theorem make use of the fact that the signals are one-dimensional; i.e., that the group
is $\Z_N$. A natural question for future work is to prove the analogue of
Patterson's theorem for any abelian group $G$ . If such a theorem held, then Theorem~\ref{thm:sparse} could be generalized to the case where $G$ is an arbitrary finite abelian group.

\subsection{Other questions} \label{sec.otherquestions}
\begin{enumerate}
  \item
In our model we assume that each atom is represented by a single letter. An interesting alternative model to investigate is to assume that the atoms are represented by a few letters placed consecutively in $[0,N-1]$. While this model is combinatorially more complicated, it may also be more likely to resolve the ambiguities of the auto-correlation.

\item Another model worth further investigation,
  particularly in higher dimensions, is to assume that the separate atoms are placed within the basic crystal structure in a regular way. For example, if $G= \Z_N^2$ and our alphabet is $\{a,b\}$ we might assume that the sets $S_a(x)$ and $S_b(x)$ are the orbits of different cyclic subgroups of $G$. (Note that $G=\Z_N^2$ has many distinct cyclic subgroups.)
  
\item In an X-ray crystallography experiment, the measurement is
  contaminated with noise, which is characterized by Poisson
  statistics. In this case, we are not searching for a signal which is
  precisely consistent with the power spectrum (as in this paper), but
  only approximately consistent. Understanding the
  information-theoretic limits of this problem, namely, what is the
  optimal expected error regardless of any specific algorithm, is an
  important research question.
\end{enumerate}
}
\section*{Acknowledgments}
The authors were supported by the BSF grant no. 2020159. T.B. is also supported in part by the NSF-BSF grant no. 2019752, and the ISF grant no. 1924/21 and D.E. was also supported by NSF-DMS 1906725 and NSF-DMS 2205626. 
{The problems posed in Section \ref{sec.otherquestions} are based on remarks of the anonymous referee.}
  

\bibliographystyle{plain}

\end{document}